\documentclass[oneside]{amsart}
\pdfoutput=1

\usepackage{tikz,tikz-3dplot}
\usetikzlibrary{positioning,arrows.meta,decorations.pathreplacing,decorations.pathmorphing,decorations.markings,intersections}
\tikzset{li/.style={line width=1pt}}

\makeatletter
 \tikzoption{canvas is plane}[]{\@setOxy#1}
 \def\@setOxy O(#1,#2,#3)x(#4,#5,#6)y(#7,#8,#9)%
   {\def\tikz@plane@origin{\pgfpointxyz{#1}{#2}{#3}}%
    \def\tikz@plane@x{\pgfpointxyz{#4}{#5}{#6}}%
    \def\tikz@plane@y{\pgfpointxyz{#7}{#8}{#9}}%
    \tikz@canvas@is@plane
   }
\makeatother

\usepackage{graphicx,amssymb,amsfonts,amsmath,amsthm,newlfont}

\newcommand{\CC}{{\mathbb C}}

\newcommand{\FF}{{\mathbb F}}

\newcommand{\RR}{{\mathbb R}}

\newcommand{\ZZ}{{\mathbb Z}}
\newcommand{\ee}{\mathrm{e}}
\newcommand{\ii}{\mathrm{i}}

\newcommand{\dd}{\mathrm{d}}
\newcommand{\Li}{\mathrm{Li}\,}

\newcommand{\pp}{{\overline{p}}}

\newcommand{\zz}{{\overline{z}}}

\newcommand{\sE}{\mathcal{E}}

\newcommand{\sV}{\mathcal{V}}

\renewcommand{\Im}{\mathop\mathrm{Im}}
\renewcommand{\Re}{\mathop\mathrm{Re}}

\theoremstyle{plain}
\newtheorem{thm}{Theorem}
\newtheorem{lem}[thm]{Lemma}

\newcommand{\thetreeX}{
\begin{minipage}{2.6cm}
\begin{tikzpicture}
      \draw [thick] (-1,0) -- (0,0);
      \draw [thick] (0,0) -- (0.5,0.866);
      \draw [thick] (0,0) -- (0.5,-0.866);
      \draw (0.7,0.866) node{${\scriptstyle 1}$};
      \draw (-1.2,0) node{${\scriptstyle 0}$};
      \draw (0.7,-0.866) node{${\scriptstyle z}$};
      \node at (0,0) [right = 0.1] {$\scriptstyle x_1$};
\end{tikzpicture}
\end{minipage}}

\newcommand{\thetreePp}{
\begin{minipage}{2.6cm}
\begin{tikzpicture}
      \draw [thick,<-] (-0.8,0.4) -- (0.2,0.95);
      \draw [thick] (-1,0) -- (0.5,0.866);
      \draw [thick] (-1,0) -- (0,0);
      \draw [thick] (0,0) -- (0.5,0.866);
      \draw [thick] (0,0) -- (0.5,-0.866);
      \draw[thick,<-, bend right=30] (0.3,-0.9) to (-0.9,-0.2);
      \draw (-0.4,0.9) node{${\scriptstyle p_2}$};
      \draw (-0.3,-0.6) node{${\scriptstyle p_1}$};
\end{tikzpicture}
\end{minipage}}

\newcommand{\thetreePx}{
\begin{minipage}{2.6cm}
\begin{tikzpicture}
      \draw [thick] (-1,0) -- (0.5,0.866);
      \draw [thick] (-1,0) -- (0,0);
      \draw [thick] (0,0) -- (0.5,0.866);
      \draw [thick] (0,0) -- (0.5,-0.866);
      \draw (0.7,0.866) node{${\scriptstyle p_1}$};
      \draw (-1.2,0) node{${\scriptstyle 0}$};
      \draw (0.7,-0.866) node{${\scriptstyle p_2}$};
\end{tikzpicture}
\end{minipage}}

\newcommand{\momentumthreepoint}{
\begin{minipage}{2.9cm}
\begin{tikzpicture}[rotate=-90]
    \def\circleradius{0.5cm}
    \def\linelength{1cm}
    \def\arcradius{0.35cm}
    \def\shrink{0.4cm} 
    \def\offset{-0.35cm} 

    \coordinate (Afull) at (120:{\circleradius+\linelength});
    \coordinate (Bfull) at (240:{\circleradius+\linelength});
    \coordinate (Cfull) at (0:{\circleradius+\linelength});
    
    \coordinate (A) at ($ (120:{\circleradius+\linelength-\shrink}) + (30:\offset) $);
    \coordinate (B) at ($ (240:{\circleradius+\linelength-\shrink}) + (330:\offset)$);
    \coordinate (C) at ($ (0:{\circleradius+\linelength-\shrink}) + (270:\offset) $);
    \coordinate (AS) at ($ (120:{\circleradius+\linelength-\shrink}) - (30:\offset) $);

    \filldraw[black!30] (0,0) circle (\circleradius);
    \draw[thick] (0,0) circle (\circleradius);
    
    \foreach \angle in {0,120,240} {
        \draw[thick] 
            (\angle:\circleradius) -- (\angle:{\circleradius + \linelength});
    }
    
    \draw[thick, ->, bend right=30] (B) to (A);

    \draw[thick, ->, bend right=30] (AS) to (C);

    \node at (0,0) {$\scriptstyle \nu_e^p$};
    \node at (C) [above right =0.5 and 0.2] {$\scriptstyle p_1$};
    \node at (A) [left = 0.4] {$\scriptstyle p_2$};
    \node at (0.8,-0.8) {$\scriptstyle G^p$};
\end{tikzpicture}
\end{minipage}}

\newcommand{\positionthreepoint}{
\begin{minipage}{2.9cm}
\begin{tikzpicture}[rotate=-90]
    \def\circleradius{0.5cm}
    \def\labels{0.8cm}

    \filldraw[black!30] (0,0) circle (\circleradius);
    \draw[thick] (0,0) circle (\circleradius);
    
    \foreach \angle in {0,120,240} {
        \filldraw
            (\angle:\circleradius) circle (0.05);
    }

    \node at (0,0) {$\scriptstyle \nu_e$};
    \node at (0:\labels) {$\scriptstyle p_2$};
    \node at (120:\labels) {$\scriptstyle 0$};
    \node at (240:\labels) {$\scriptstyle p_1$};
    \node at (0.8,0.8) {$\scriptstyle G$};
\end{tikzpicture}
\end{minipage}}

\newcommand{\momentumthreepointza}{
\begin{minipage}{2cm}
\begin{tikzpicture}[rotate=-90]
    \def\circleradius{0.5cm}
    \def\linelength{1cm}
    \def\arcradius{0.35cm}
    \def\shrink{0.4cm} 
    \def\offset{-0.35cm} 

    \coordinate (Afull) at (120:{\circleradius+\linelength});
    \coordinate (Bfull) at (240:{\circleradius+\linelength});
    \coordinate (Cfull) at (0:{\circleradius+\linelength});
    
    \coordinate (A) at ($ (120:{\circleradius+\linelength-\shrink}) + (30:\offset) $);
    \coordinate (B) at ($ (240:{\circleradius+\linelength-\shrink}) + (330:\offset)$);
    \coordinate (C) at ($ (0:{\circleradius+\linelength-\shrink}) + (270:\offset) $);
    \coordinate (AS) at ($ (120:{\circleradius+\linelength-\shrink}) - (30:\offset) $);

    \filldraw[black!30] (0,0) circle (\circleradius);
    \draw[thick] (0,0) circle (\circleradius);
    
    \foreach \angle in {0,120,240} {
        \filldraw
            (\angle:\circleradius) circle (0.05);
    }

    \draw[thick, ->, bend right=30] (B) to (A);

    \draw[thick, ->, bend right=30] (AS) to (C);

    \node at (0,0) {$G^p$};
    \node at (C) [above right =0.5 and 0.2] {$\scriptstyle z_1$};
    \node at (A) [left = 0.4] {$\scriptstyle z_2$};
\end{tikzpicture}
\end{minipage}}

\newcommand{\momentumthreepointzb}{
\begin{minipage}{2.2cm}
\begin{tikzpicture}[rotate=-90]
    \def\circleradius{0.5cm}
    \def\linelength{1cm}
    \def\arcradius{0.35cm}
    \def\shrink{0.4cm} 
    \def\offset{-0.35cm} 

    \coordinate (Afull) at (120:{\circleradius+\linelength});
    \coordinate (Bfull) at (240:{\circleradius+\linelength});
    \coordinate (Cfull) at (0:{\circleradius+\linelength});
    
    \coordinate (A) at ($ (120:{\circleradius+\linelength-\shrink}) + (30:\offset) $);
    \coordinate (B) at ($ (240:{\circleradius+\linelength-\shrink}) + (330:\offset)$);
    \coordinate (C) at ($ (0:{\circleradius+\linelength-\shrink}) + (270:\offset) $);
    \coordinate (AS) at ($ (120:{\circleradius+\linelength-\shrink}) - (30:\offset) $);

    \filldraw[black!30] (0,0) circle (\circleradius);
    \draw[thick] (0,0) circle (\circleradius);
    
    \draw[thick] (120:\circleradius) -- (120:{\circleradius + \linelength});
    
    \node at (110:{\circleradius + \linelength/2}) {$\scriptstyle \nu^p$};

    \foreach \angle in {0,240} {
        \filldraw
            (\angle:\circleradius) circle (0.05);
    }

    \draw[thick, ->, bend right=30] (B) to (A);

    \draw[thick, ->, bend right=30] (AS) to (C);

    \node at (0,0) {$G^p$};
    \node at (C) [above right =0.5 and 0.2] {$\scriptstyle z_1$};
    \node at (A) [left = 0.4] {$\scriptstyle z_2$};
\end{tikzpicture}
\end{minipage}}

\newcommand{\momentumthreepointzc}{
\begin{minipage}{2.7cm}
\begin{tikzpicture}[rotate=-90]
    \def\circleradius{0.5cm}
    \def\linelength{1cm}
    \def\arcradius{0.35cm}
    \def\shrink{0.4cm} 
    \def\offset{-0.35cm} 

    \coordinate (Afull) at (120:{\circleradius+\linelength});
    \coordinate (Bfull) at (240:{\circleradius+\linelength});
    \coordinate (Cfull) at (0:{\circleradius+\linelength});
    
    \coordinate (A) at ($ (120:{\circleradius+\linelength-\shrink}) + (30:\offset) $);
    \coordinate (B) at ($ (240:{\circleradius+\linelength-\shrink}) + (330:\offset)$);
    \coordinate (C) at ($ (0:{\circleradius+\linelength-\shrink}) + (270:\offset) $);
    \coordinate (AS) at ($ (120:{\circleradius+\linelength-\shrink}) - (30:\offset) $);

    \filldraw[black!30] (0,0) circle (\circleradius);
    \draw[thick] (0,0) circle (\circleradius);
    \draw[thick] (240:\circleradius) -- (240:{\circleradius + \linelength});
    \node at (252:{\circleradius + \linelength/2}) {$\scriptstyle \nu^p$};

    \foreach \angle in {0,120} {
        \filldraw
            (\angle:\circleradius) circle (0.05);
    }

    \draw[thick, ->, bend right=30] (B) to (A);

    \draw[thick, ->, bend right=30] (AS) to (C);

    \node at (0,0) {$G^p$};
    \node at (C) [above right =0.5 and 0.2] {$\scriptstyle z_1$};
    \node at (A) [left = 0.4] {$\scriptstyle z_2$};
\end{tikzpicture}
\end{minipage}}

\newcommand{\momentumthreepointzd}{
\begin{minipage}{2cm}
\begin{tikzpicture}[rotate=-90]
    \def\circleradius{0.5cm}
    \def\linelength{1cm}
    \def\arcradius{0.35cm}
    \def\shrink{0.4cm} 
    \def\offset{-0.35cm} 

    \coordinate (Afull) at (120:{\circleradius+\linelength});
    \coordinate (Bfull) at (240:{\circleradius+\linelength});
    \coordinate (Cfull) at (0:{\circleradius+\linelength});
    
    \coordinate (A) at ($ (120:{\circleradius+\linelength-\shrink}) + (30:\offset) $);
    \coordinate (B) at ($ (240:{\circleradius+\linelength-\shrink}) + (330:\offset)$);
    \coordinate (C) at ($ (0:{\circleradius+\linelength-\shrink}) + (270:\offset) $);
    \coordinate (AS) at ($ (120:{\circleradius+\linelength-\shrink}) - (30:\offset) $);

    \filldraw[black!30] (0,0) circle (\circleradius);
    \draw[thick] (0,0) circle (\circleradius);
    \draw[thick] (0:\circleradius) -- (0:{\circleradius + \linelength});
    \node at (-12:{\circleradius + \linelength/2}) {$\scriptstyle \nu^p$};

    \foreach \angle in {120,240} {
        \filldraw
            (\angle:\circleradius) circle (0.05);
    }

    \draw[thick, ->, bend right=30] (B) to (A);

    \draw[thick, ->, bend right=30] (AS) to (C);

    \node at (0,0) {$G^p$};
    \node at (C) [above right =0.5 and 0.2] {$\scriptstyle z_1$};
    \node at (A) [left = 0.4] {$\scriptstyle z_2$};
\end{tikzpicture}
\end{minipage}}

\newcommand{\positionthreepointzb}{
\begin{minipage}{4cm}
\begin{tikzpicture}[rotate=-90]
    \def\circleradius{0.5cm}
    \def\linelength{1cm}
    \def\labels{0.8cm}

    \coordinate (Afull) at (120:{\circleradius+\linelength});
    \coordinate (Bfull) at (240:{\circleradius+\linelength});
    \coordinate (Cfull) at (0:{\circleradius+\linelength});
    
    \filldraw[black!30] (0,0) circle (\circleradius);
    \draw[thick] (0,0) circle (\circleradius);
    \draw[thick] (120:\circleradius) -- (120:{\circleradius + \linelength});
    \draw[thick] (240:\circleradius) .. controls (280:\linelength) and (320:\linelength) .. (0:\circleradius);
    \node at (110:\circleradius + \linelength/2) {$\scriptstyle \nu$};
    \node at (300:1.23cm) {$\scriptstyle \nu-\frac{\lambda+1}\lambda$};

    \foreach \angle in {0,240} {
        \filldraw
            (\angle:\circleradius) circle (0.05);
    }
    \filldraw(120:\circleradius+\linelength) circle (0.05);

    \node at (0,0) {$G$};
    \node at (0:\labels) {$\scriptstyle z_2$};
    \node at (120:\labels + \linelength) {$\scriptstyle 0$};
    \node at (240:\labels) {$\scriptstyle z_1$};

\end{tikzpicture}
\end{minipage}}

\newcommand{\positionthreepointzc}{
\begin{minipage}{4cm}
\begin{tikzpicture}[rotate=-90]
    \def\circleradius{0.5cm}
    \def\linelength{1cm}
    \def\labels{0.8cm}

    \coordinate (Afull) at (120:{\circleradius+\linelength});
    \coordinate (Bfull) at (240:{\circleradius+\linelength});
    \coordinate (Cfull) at (0:{\circleradius+\linelength});
    
    \filldraw[black!30] (0,0) circle (\circleradius);
    \draw[thick] (0,0) circle (\circleradius);
    \draw[thick] (240:\circleradius) -- (240:{\circleradius + \linelength});
    \draw[thick] (0:\circleradius) .. controls (40:\linelength) and (80:\linelength) .. (120:\circleradius);
    \node at (250:\circleradius + \linelength/2) {$\scriptstyle \nu$};
    \node at (60:1.23cm) {$\scriptstyle \nu-\frac{\lambda+1}\lambda$};

    \foreach \angle in {0,120} {
        \filldraw
            (\angle:\circleradius) circle (0.05);
    }
    \filldraw(240:\circleradius+\linelength) circle (0.05);

    \node at (0,0) {$G$};
    \node at (0:\labels) {$\scriptstyle z_2$};
    \node at (120:\labels) {$\scriptstyle 0$};
    \node at (240:\labels + \linelength) {$\scriptstyle z_1$};

\end{tikzpicture}
\end{minipage}}

\newcommand{\positionthreepointzd}{
\begin{minipage}{2.5cm}
\begin{tikzpicture}[rotate=-90]
    \def\circleradius{0.5cm}
    \def\linelength{1cm}
    \def\labels{0.8cm}

    \coordinate (Afull) at (120:{\circleradius+\linelength});
    \coordinate (Bfull) at (240:{\circleradius+\linelength});
    \coordinate (Cfull) at (0:{\circleradius+\linelength});
    
    \filldraw[black!30] (0,0) circle (\circleradius);
    \draw[thick] (0,0) circle (\circleradius);
    \draw[thick] (0:\circleradius) -- (0:{\circleradius + \linelength});
    \draw[thick] (120:\circleradius) .. controls (160:\linelength) and (200:\linelength) .. (240:\circleradius);
    \node at (-10:\circleradius + \linelength/2) {$\scriptstyle \nu$};
    \node at (180:1cm) {$\scriptstyle \nu-\frac{\lambda+1}\lambda$};

    \foreach \angle in {120,240} {
        \filldraw
            (\angle:\circleradius) circle (0.05);
    }
    \filldraw(0:\circleradius+\linelength) circle (0.05);

    \node at (0,0) {$G$};
    \node at (0:\labels + \linelength) {$\scriptstyle z_2$};
    \node at (120:\labels) {$\scriptstyle 0$};
    \node at (240:\labels) {$\scriptstyle z_1$};

\end{tikzpicture}
\end{minipage}}

\newcommand{\graphicalfunctiona}{
\begin{minipage}{2.5cm}
\begin{tikzpicture}[rotate=-90]
    \def\circleradius{0.5cm}
    \def\linelength{1cm}
    \def\labels{0.8cm}

    \coordinate (Afull) at (120:{\circleradius+\linelength});
    \coordinate (Bfull) at (240:{\circleradius+\linelength});
    \coordinate (Cfull) at (0:{\circleradius+\linelength});
    
    \filldraw[black!30] (0,0) circle (\circleradius);
    \draw[thick] (0,0) circle (\circleradius);
    \draw[thick] (0:\circleradius) -- (0:{\circleradius + \linelength});
    \node at (-10:\circleradius + \linelength/2) {$\scriptstyle \nu$};

    \foreach \angle in {120,240} {
        \filldraw
            (\angle:\circleradius) circle (0.05);
    }
    \filldraw(0:\circleradius+\linelength) circle (0.05);

    \node at (180:\linelength) {$G$};
    \node at (16:1.65cm) {$\scriptstyle z_2=z$};
    \node at (120:\labels) {$\scriptstyle 0$};
    \node at (240:\labels+0.1cm) {$\scriptstyle z_1=1$};

\end{tikzpicture}
\end{minipage}}

\newcommand{\graphicalfunctionb}{
\begin{minipage}{4cm}
\begin{tikzpicture}[rotate=-90]
    \def\circleradius{0.5cm}
    \def\linelength{1cm}
    \def\labels{0.8cm}

    \coordinate (Afull) at (120:{\circleradius+\linelength});
    \coordinate (Bfull) at (240:{\circleradius+\linelength});
    \coordinate (Cfull) at (0:{\circleradius+\linelength});
    
    \filldraw[black!30] (0,0) circle (\circleradius);
    \draw[thick] (0,0) circle (\circleradius);
    \draw[thick] (120:\circleradius) -- (120:{\circleradius + \linelength});
    \draw[thick] (240:\circleradius) .. controls (280:\linelength) and (320:\linelength) .. (0:\circleradius);
    \draw[thick] (120:\circleradius) .. controls (160:\linelength) and (200:\linelength) .. (240:\circleradius);
    \draw[thick] (120:\circleradius+\linelength) .. controls (160:2cm) and (220:2cm) .. (240:\circleradius);
    \node at (110:\circleradius + \linelength/2) {$\scriptstyle \nu$};
    \node at (300:1.23cm) {$\scriptstyle \nu-\frac{\lambda+1}\lambda$};
    \node at (170:1cm) {$\scriptstyle \frac{\lambda+1}\lambda-N_G$};
    \node at (220:1.6cm) {$\scriptstyle N_G-\nu$};
    
    \foreach \angle in {0,240} {
        \filldraw
            (\angle:\circleradius) circle (0.05);
    }
    \filldraw(120:\circleradius+\linelength) circle (0.05);
   \filldraw(0:\circleradius+2*\linelength) circle (0.0);

    \node at (0:\labels) {$\scriptstyle z_2=z$};
    \node at (120:\labels + \linelength) {$\scriptstyle 0$};
    \node at (250:\labels+0.15cm) {$\scriptstyle z_1=1$};

\end{tikzpicture}
\end{minipage}}

\newcommand{\graphicalfunctionc}{
\begin{minipage}{4cm}
\begin{tikzpicture}[rotate=-90]
    \def\circleradius{0.5cm}
    \def\linelength{1cm}
    \def\labels{0.8cm}

    \coordinate (Afull) at (120:{\circleradius+\linelength});
    \coordinate (Bfull) at (240:{\circleradius+\linelength});
    \coordinate (Cfull) at (0:{\circleradius+\linelength});
    
    \filldraw[black!30] (0,0) circle (\circleradius);
    \draw[thick] (0,0) circle (\circleradius);
    \draw[thick] (240:\circleradius) -- (240:{\circleradius + \linelength});
    \draw[thick] (0:\circleradius) .. controls (40:\linelength) and (80:\linelength) .. (120:\circleradius);
    \draw[thick] (120:\circleradius) .. controls (160:\linelength) and (200:\linelength) .. (240:\circleradius);
    \draw[thick] (120:\circleradius) .. controls (140:2cm) and (200:2cm) .. (240:\circleradius+\linelength);
    \node at (250:\circleradius + \linelength/2) {$\scriptstyle \nu$};
    \node at (60:1.23cm) {$\scriptstyle \nu-\frac{\lambda+1}\lambda$};
    \node at (190:1cm) {$\scriptstyle \frac{\lambda+1}\lambda-N_G$};
    \node at (140:1.6cm) {$\scriptstyle N_G-\nu$};

    \foreach \angle in {0,120} {
        \filldraw
            (\angle:\circleradius) circle (0.05);
    }
    \filldraw(240:\circleradius+\linelength) circle (0.05);
   \filldraw(0:\circleradius+2*\linelength) circle (0.0);

    \node at (0:\labels) {$\scriptstyle z_2=z$};
    \node at (117:\labels) {$\scriptstyle 0$};
    \node at (240:\labels + \linelength+0.1cm) {$\scriptstyle z_1=1$};

\end{tikzpicture}
\end{minipage}}

\newcommand{\grapha}{
\begin{minipage}{2.5cm}
\begin{tikzpicture}
    \def\linelength{0.8cm}
    \coordinate (A) at (-90:\linelength);
    \coordinate (B) at (30:\linelength);
    \coordinate (C) at (150:\linelength);
    \coordinate (D) at (-0.866*\linelength,1.5*\linelength);
    \coordinate (O) at (0,0);
    
    \foreach \pos in {C,O} {
        \filldraw
            (\pos) circle (0.05);
    }
    \draw[thick] (A) -- (O) -- (B) -- (C) -- (D);
    \draw[thick] (O) -- (C);
    \filldraw (A) circle (0.05) node[anchor=north] {$\scriptstyle z$};
    \filldraw (B) circle (0.05) node[anchor=west] {$\scriptstyle 0$};
    \filldraw (D) circle (0.05) node[anchor=south] {$\scriptstyle 1$};
\end{tikzpicture}
\end{minipage}}

\newcommand{\graphb}{
\begin{minipage}{3.7cm}
\begin{tikzpicture}
    \def\linelength{0.8cm}
    \coordinate (A) at (-90:\linelength);
    \coordinate (B) at (30:\linelength);
    \coordinate (C) at (150:\linelength);
    \coordinate (D) at (-0.866*\linelength,1.5*\linelength);
    \coordinate (E) at (1.866*\linelength,0.5*\linelength);
    \coordinate (O) at (0,0);
    
    \foreach \pos in {B,C} {
        \filldraw
            (\pos) circle (0.05);
    }
    \draw[thick] (C) -- (O) -- (B) -- (C) -- (D) -- (B) -- (E);
    \filldraw (O) circle (0.05) node[anchor=north] {$\scriptstyle z$};
    \filldraw (E) circle (0.05) node[anchor=west] {$\scriptstyle 0$};
    \filldraw (D) circle (0.05) node[anchor=south] {$\scriptstyle 1$};
    \draw[thick] (O) .. controls (200:2*\linelength) and (150:2*\linelength) .. (D);
    \node at (-1.7*\linelength,\linelength) {$\scriptstyle -1$};

\end{tikzpicture}
\end{minipage}}

\newcommand{\graphc}{
\begin{minipage}{2.5cm}
\begin{tikzpicture}
    \def\linelength{0.8cm}
    \coordinate (A) at (-90:\linelength);
    \coordinate (B) at (30:\linelength);
    \coordinate (C) at (150:\linelength);
    \coordinate (D) at (-0.866*\linelength,1.5*\linelength);
    \coordinate (E) at (1.866*\linelength,0.5*\linelength);
    \coordinate (F) at (-0.866*\linelength,2.5*\linelength);
    \coordinate (O) at (0,0);
    
    \foreach \pos in {C,D} {
        \filldraw
            (\pos) circle (0.05);
    }
    \draw[thick] (O) -- (C) -- (B) -- (D) -- (C);
    \draw[thick] (D) -- (F);
    \filldraw (O) circle (0.05) node[anchor=north] {$\scriptstyle z$};
    \filldraw (B) circle (0.05) node[anchor=west] {$\scriptstyle 0$};
    \filldraw (F) circle (0.05) node[anchor=south] {$\scriptstyle 1$};

\end{tikzpicture}
\end{minipage}}

\newcommand{\graphd}{
\begin{minipage}{2cm}
\begin{tikzpicture}
    \def\linelength{0.8cm}
    \coordinate (A) at (-90:\linelength);
    \coordinate (B) at (30:\linelength);
    \coordinate (C) at (150:\linelength);
    \coordinate (D) at (-0.866*\linelength,1.5*\linelength);
    \coordinate (E) at (1.866*\linelength,0.5*\linelength);
    \coordinate (F) at (-0.866*\linelength,2.5*\linelength);
    \coordinate (G) at (90:\linelength);
    \coordinate (O) at (0,0);
    
    \foreach \pos in {A,B,C,D,G,O} {
        \filldraw
            (\pos) circle (0.05);
    }
    \draw[thick] (A) -- (O) -- (B) -- (G) -- (O) -- (C) -- (G);
    \draw[thick] (C) -- (D);

\end{tikzpicture}
\end{minipage}}

\newcommand{\graphe}{
\begin{minipage}{3.5cm}
\begin{tikzpicture}
    \def\linelength{0.8cm}
    \coordinate (A) at (-90:\linelength);
    \coordinate (B) at (30:\linelength);
    \coordinate (C) at (150:\linelength);
    \coordinate (D) at (-0.866*\linelength,1.5*\linelength);
    \coordinate (E) at (1.866*\linelength,0.5*\linelength);
    \coordinate (F) at (-0.866*\linelength,2.5*\linelength);
    \coordinate (G) at (90:\linelength);
    \coordinate (H) at (1.866*\linelength,0.5*\linelength);
    \coordinate (O) at (0,0);
    
    \foreach \pos in {B,C,D,G,O,H} {
        \filldraw
            (\pos) circle (0.05);
    }
    \draw[thick] (O) -- (B) -- (G) -- (O) -- (C) -- (G);
    \draw[thick] (C) -- (D);
    \draw[thick] (B) -- (H);

    \draw[thick] (D) .. controls (140:1.6cm) and (200:1.4cm) .. (O);
    \draw[thick] (B) .. controls (50:1.1cm) and (90:1.4cm) .. (D);

    \node at (-1.3cm,0.5cm) {$\scriptstyle -1$};

\end{tikzpicture}
\end{minipage}}

\newcommand{\graphf}{
\begin{minipage}{2.5cm}
\begin{tikzpicture}
    \def\linelength{0.8cm}
    \coordinate (A) at (-90:\linelength);
    \coordinate (B) at (30:\linelength);
    \coordinate (C) at (150:\linelength);
    \coordinate (D) at (-0.866*\linelength,1.5*\linelength);
    \coordinate (E) at (1.866*\linelength,0.5*\linelength);
    \coordinate (F) at (-0.866*\linelength,2.5*\linelength);
    \coordinate (G) at (90:\linelength);
    \coordinate (H) at (1.866*\linelength,0.5*\linelength);
    \coordinate (O) at (0,0);
    
    \foreach \pos in {B,C,D,F,G,O} {
        \filldraw
            (\pos) circle (0.05);
    }
    \draw[thick] (B) -- (G) -- (C) -- (O) -- (G);
    \draw[thick] (C) -- (D) -- (F);

    \draw[thick] (B) .. controls (50:1.1cm) and (90:1.4cm) .. (D);

\end{tikzpicture}
\end{minipage}}

\newcommand{\graphg}{
\begin{minipage}{2.5cm}
\begin{tikzpicture}
    \def\linelength{0.8cm}
    \coordinate (A) at (-90:\linelength);
    \coordinate (B) at (30:\linelength);
    \coordinate (C) at (150:\linelength);
    \coordinate (D) at (-0.866*\linelength,1.5*\linelength);
    \coordinate (E) at (1.866*\linelength,0.5*\linelength);
    \coordinate (F) at (-0.866*\linelength,2.5*\linelength);
    \coordinate (G) at (90:\linelength);
    \coordinate (H) at (1.866*\linelength,0.5*\linelength);
    \coordinate (I) at (-30:\linelength);
    \coordinate (J) at (210:\linelength);
    \coordinate (K) at (-1.866*\linelength,-\linelength);
    \coordinate (O) at (0,0);
    
    \draw[thick] (G) -- (I);
    \draw[white,line width=5pt] (O) -- (B);

    \foreach \pos in {A,O,B,C,G,I,J,K,O} {
        \filldraw
            (\pos) circle (0.05);
    }
    \draw[thick] (A) -- (I) -- (B) -- (G);
    \draw[thick] (I) -- (O) -- (A) -- (K) -- (J) -- (C) -- (G);
    \draw[thick] (K) -- (C);
    \draw[thick] (A) -- (J) -- (O) -- (B);

    \node at (-0.3cm,-1.5cm) {$P_{7,1}$};

\end{tikzpicture}
\end{minipage}}

\newcommand{\graphh}{
\begin{minipage}{2.5cm}
\begin{tikzpicture}
    \def\linelength{0.8cm}
    \coordinate (A) at (-90:\linelength);
    \coordinate (B) at (30:\linelength);
    \coordinate (C) at (150:\linelength);
    \coordinate (D) at (-0.866*\linelength,1.5*\linelength);
    \coordinate (E) at (1.866*\linelength,0.5*\linelength);
    \coordinate (F) at (-0.866*\linelength,2.5*\linelength);
    \coordinate (G) at (90:\linelength);
    \coordinate (H) at (1.866*\linelength,0.5*\linelength);
    \coordinate (I) at (-30:\linelength);
    \coordinate (J) at (210:\linelength);
    \coordinate (K) at (-1.866*\linelength,-\linelength);
    \coordinate (L) at (-1.4*\linelength,-0.2cm);
    \coordinate (M) at (-\linelength,-\linelength);
    \coordinate (O) at (0,0);
    
    \draw[thick] (G) -- (I);
    \draw[white,line width=5pt] (O) -- (B);
    \draw[thick] (A) -- (C);
    \draw[white,line width=5pt] (L) -- (O) -- (M);

    \foreach \pos in {A,O,B,C,G,I,L,M,O} {
        \filldraw
            (\pos) circle (0.05);
    }
    \draw[thick] (A) -- (I) -- (B) -- (G);
    \draw[thick] (I) -- (O) -- (A) -- (M) -- (L) -- (C) -- (G);
    \draw[thick] (M) -- (O) -- (L);
    \draw[thick] (O) -- (B);

    \node at (0cm,-1.5cm) {$P^{\mathrm{non\,}\phi^4}_{7,31}$};

\end{tikzpicture}
\end{minipage}}

\newcommand{\graphi}{
\begin{minipage}{3cm}
\begin{tikzpicture}
    \def\linelength{0.8cm}
    \coordinate (A) at (-90:\linelength);
    \coordinate (B) at (30:\linelength);
    \coordinate (C) at (150:1.2*\linelength);
    \coordinate (D) at (-0.866*\linelength,1.5*\linelength);
    \coordinate (E) at (1.866*\linelength,0.5*\linelength);
    \coordinate (F) at (-0.866*\linelength,2.5*\linelength);
    \coordinate (G) at (90:\linelength);
    \coordinate (H) at (1.866*\linelength,0.5*\linelength);
    \coordinate (I) at (-30:\linelength);
    \coordinate (J) at (210:\linelength);
    \coordinate (K) at (-1.866*\linelength,-\linelength);
    \coordinate (L) at (-1.4*\linelength,-0.2cm);
    \coordinate (M) at (-\linelength,-\linelength);
    \coordinate (O) at (0,0);
    
    \draw[thick] (G) -- (I);
    \draw[white,line width=5pt] (O) -- (B);
    \draw[thick] (M) .. controls (-2cm,-0.4cm) and (-2cm,1cm) .. (G);

    \foreach \pos in {A,O,B,C,G,I,L,M,O} {
        \filldraw
            (\pos) circle (0.05);
    }
    \draw[thick] (A) -- (I) -- (B) -- (G);
    \draw[thick] (I) -- (O) -- (A) -- (M) -- (L) -- (C) -- (G) -- (O);
    \draw[thick] (M) -- (O) -- (L);
    \draw[thick] (C) -- (O) -- (B);

    \node at (-0.2cm,0.42cm) {$\scriptstyle -1$};
    \node at (0cm,-1.5cm) {$P^{\mathrm{non\,}\phi^4}_{7,18}$};

\end{tikzpicture}
\end{minipage}}

\title{Self-duality of massless scalar three-point amplitudes}
\date{}
\author{Oliver Schnetz}
\address{Oliver Schnetz\\
II. Institut f\"ur Theoretische Physik\\
Luruper Chaussee 149\\
22761 Hamburg, Germany}
\email{schnetz@mi.uni-erlangen.de}

\begin{document}

\begin{abstract}
We prove that off-shell massless scalar three-point Feynman integrals are self-dual under Fourier transformation. This implies that a momentum space integral can be expressed
as the position space integral of the same Feynman graph with transformed edge-weights (not the dual graph) if external vertices are labeled accordingly.
In particular, any off-shell massless scalar three-point Feynman integral can be expressed as a graphical function.
The result follows immediately from a theorem by M. Golz, E. Panzer and the author on parametric representations of position space integrals (2015),
but it was only observed by X. Jiang in 2025 in the context of four-dimensional $\mathcal{N}=4$ Super-Yang-Mills theory.
We generalize Jiang's result and discuss the consequences of the self-duality in the context of graphical functions.
In particular, we derive a new identity for graphical functions and a new twist relation for scalar integrals (Feynman periods) in $\phi^4$ theory.
\end{abstract}
\maketitle

\section{Introduction}
We consider three-point functions in massless scalar (spin zero) quantum field theories (QFT). If the QFT has a three-valent vertex (like six-dimensional $\phi^3$ theory), then
the three-point function has the intrinsic meaning as a vertex correction. Otherwise, three-point functions can be used as a tool to calculate scalar quantities like renormalization
functions (the $\beta$ and $\gamma$ functions). This method was used in position space where the three-point function naturally gives rise to a function on the complex plane:
the graphical function \cite{5lphi3,gfe,recursive,gf,numfunct,7loops,6loops,gft}.

Concretely, we consider a connected graph $G$ in a massless scalar QFT in (possibly noninteger) dimension
\begin{equation}
D=2\lambda+2.
\end{equation}
The graph $G$ has three external vertices $z_0,z_1,z_2\in\RR^D$ and $|\sV_G^{\mathrm{int}}|$ internal vertices $x_1,\ldots,x_{|\sV_G^{\mathrm{int}}|}$.

The edges of $G$ have weights $\nu_e\in\RR$, so that the position space Feynman propagator of the edge $e=xy$ with vertices $x,y\in\RR^D$ has the form
\begin{equation}\label{eqpe}
p_e(x,y)=\frac1{||x-y||^{2\lambda\nu_e}}.
\end{equation}
A standard edge $e$ has weight $\nu_e=1$. The introduction of arbitrary weights gives us the freedom to consider more general configurations. For example,
we can write a double edge (a bubble) as a single edge of weight $\nu_e=2$. Integrating out two-point insertions gives rise to noninteger edge weights.

The signature of the norm $||\cdot||$ will not be important in the following, so that we restrict ourselves to Euclidean signature.
The three-point integral $A_G$ is defined by integration over the internal vertices,
\begin{equation}\label{AG}
A_G(z_0,z_1,z_2)=\int_{\RR^{D|\sV_G^{\mathrm{int}}|}}\Big(\prod_{x_i\in\sV_G^{\mathrm{int}}}\frac{\dd^Dx_i}{\pi^{D/2}}\Big)\prod_{e\in\sE_G}p_e(x,z),
\end{equation}
where the propagator $p_e$ depends on $z=(z_1,z_2,z_3)$ if $e$ connects to an external vertex.

In position space, the scaling weight of the graph $G$ (the superficial degree of divergence) is
\begin{equation}\label{eqNG}
N_G=\Big(\sum_{e\in\sE_G}\nu_e\Big)-\frac{D}{2\lambda}|\sV_G^{\mathrm{int}}|.
\end{equation}
A key feature of position space three-point integrals is that they can be expressed as a single-valued function on the punctured complex plane $\CC\backslash\{0,1\}$ \cite{par},

\begin{equation}\label{AGf}
A_G(z_0,z_1,z_2)=||z_1-z_0||^{-2\lambda N_G}f_G^{(\lambda)}(z),
\end{equation}
where $z$ and its complex conjugate $\zz$ relate to the invariants $||z_{ij}||^2=||z_i-z_j||^2$ according to
\begin{equation}\label{inv3}
\frac{||z_{20}||^2}{||z_{10}||^2}=z\zz,\qquad \frac{||z_{21}||^2}{||z_{10}||^2}=(z-1)(\zz-1).
\end{equation}
The first nontrivial example of a graphical function is the three-star
\begin{equation}\label{3star}
\thetreeX=\frac{4\ii D(z)}{z-\zz},
\end{equation}
where $D(z)$ is the Bloch-Wigner dilogarithm \cite{Zagierdilog},
\begin{equation}\label{BWLi}
D(z)=\mathrm{Im}\,(\Li_2(z)+\log(1-z)\log|z|).
\end{equation}
Note that in (\ref{3star}) we used the labels $0,1,z\in\CC$ for the external vertices $z_0,z_1,z_2\in\RR^D$, respectively. This correspondence is implied by (\ref{inv3}) and makes
the natural embedding of $\CC$ into $\RR^D$ explicit (see Figure \ref{fig:gfC}).

\begin{figure}
\tdplotsetmaincoords{80}{120}
\centering
\begin{tikzpicture}[tdplot_main_coords,scale=.7]
\draw[thin, ->,black!50] (0,0,0) -- (4.5,0,0) node[anchor=south east,opacity=1]{$x^D$};

\draw[thin, ->,black!50] (0,0,0) -- (0,4.5,0)  node[anchor=south west,opacity=1]{$x^1$};

\draw[thin, ->,black!50] (0,0,0) -- (0,0,4.5)  node[anchor=south west,opacity=1]{$x^2$};

\tdplotsetrotatedcoords{0}{90}{0};
\draw[dotted,black!50,tdplot_rotated_coords] (0,.8,0) arc (90:180:.8);

\pgfmathsetmacro{\Ox}{32}
\pgfmathsetmacro{\Oy}{14}
\pgfmathsetmacro{\Oz}{3}
\pgfmathsetmacro{\Onex}{-5}
\pgfmathsetmacro{\Oney}{3}
\pgfmathsetmacro{\Onez}{1}
\pgfmathsetmacro{\Zx}{-12}
\pgfmathsetmacro{\Zy}{1}
\pgfmathsetmacro{\Zz}{6}

\tdplotcrossprod(\Onex,\Oney,\Onez)(\Zx,\Zy,\Zz)
\pgfmathsetmacro{\rx}{\tdplotresx}
\pgfmathsetmacro{\ry}{\tdplotresy}
\pgfmathsetmacro{\rz}{\tdplotresz}
\tdplotcrossprod(\rx,\ry,\rz)(\Onex,\Oney,\Onez)
\pgfmathsetmacro{\tx}{\tdplotresx/100}
\pgfmathsetmacro{\ty}{\tdplotresy/100}
\pgfmathsetmacro{\tz}{\tdplotresz/100}

\pgfmathsetmacro{\dplane}{\rx*\Ox + \ry*\Oy + \rz*\Oz}

\pgfmathsetmacro{\OneLen}{sqrt(\Onex*\Onex+\Oney*\Oney+\Onez*\Onez)};
\pgfmathsetmacro{\iLen}{sqrt(\tx*\tx+\ty*\ty+\tz*\tz)};

\tikzset{perspective/.style= {canvas is plane={O(0,0,0)x(\Onex/\OneLen,\Oney/\OneLen,\Onez/\OneLen)y(\tx/\iLen,\ty/\iLen,\tz/\iLen)}} }
\pgfmathsetmacro{\axisscale}{2};
\tikzset{perspective2/.style= {canvas is plane={O(0,0,0)x(\axisscale*\Onex/\OneLen,\axisscale*\Oney/\OneLen,\axisscale*\Onez/\OneLen)y(\axisscale*\tx/\iLen,\axisscale*\ty/\iLen,\axisscale*\tz/\iLen)}} }

\pgfmathsetmacro{\axisovershoot}{2};
\pgfmathsetmacro{\axisundershoot}{.5};

\coordinate (v0) at (\Ox,\Oy,\Oz);
\coordinate (v1) at ([shift={(\Onex,\Oney,\Onez)}]v0);
\coordinate (vz) at ([shift={(\Zx,\Zy,\Zz)}]v0);
\coordinate (vi) at ([shift={(\tx,\ty,\tz)}]v0);

\filldraw[black!50] (0,0,0) circle (1.3pt);

\draw[dashed,-{Stealth[length=10pt, width=15pt]},black!50] (0,0,0) -- node[inner sep=1.5pt,below right] {$z_0$} (v0);
\draw[dashed,-{Stealth[length=10pt, width=15pt]},black!50] (0,0,0) -- node[inner sep=1.5pt,below left] {$z_1$} (v1);
\draw[dashed,-{Stealth[length=10pt, width=15pt]},black!50] (0,0,0) -- node[inner sep=1.5pt,below right] {$z_2$} (vz);
\draw[thin,->] ($(v0)!-\axisundershoot/\OneLen!(v1)$) -- ($(v0)!{1+2.7*\axisovershoot/\OneLen}!(v1)$) node[perspective,anchor=north west,opacity=1]{$\Re z$};
\draw[thin,->] ($(v0)!-\axisundershoot/\iLen!(vi)$) -- ($(v0)!{\OneLen/\iLen+1.4*\axisovershoot/\iLen}!(vi)$) node[perspective,anchor=south west,opacity=1]{$\Im z$};

\coordinate (C1) at ($(v0)!{2*\OneLen/\iLen}!(vi)$);
\node[perspective2] (C) at ($(C1)!{.5}!(v1)$) {$\CC$};

\pgfmathsetmacro{\xscale}{1.7};
\coordinate (R0) at (${1 + 2.1*\axisundershoot/\OneLen + 2*\axisundershoot/\iLen}*(v0) - 2.1*\axisundershoot/\OneLen*(v1)- 2*\axisundershoot/\iLen*(vi)$);
\coordinate (R1) at (${ - \xscale*2*\axisovershoot/\OneLen + 2*\axisundershoot/\iLen}*(v0) + {1 + 2*\xscale*\axisovershoot/\OneLen}*(v1)- 2*\axisundershoot/\iLen*(vi)$);
\coordinate (R2) at (${ - \OneLen/\iLen - 2*\xscale*\axisovershoot/\OneLen - 2*\axisovershoot/\iLen}*(v0) + {1 + 2*\xscale*\axisovershoot/\OneLen}*(v1)+ {\OneLen/\iLen + 2*\axisovershoot/\iLen}*(vi)$);
\coordinate (R3) at (${1 - \OneLen/\iLen + 2.1*\axisundershoot/\OneLen - 2*\axisovershoot/\iLen}*(v0) - {2.1*\axisundershoot/\OneLen}*(v1)+ {\OneLen/\iLen + 2*\axisovershoot/\iLen}*(vi)$);

\draw[opacity = .05,fill,black!50]  (R0) -- (R1) -- (R2) -- (R3);

\filldraw (v0) circle(1.3pt) node[below left,perspective] {$0$};
\filldraw (v1) circle(1.3pt) node[below,perspective] {$1$};
\filldraw (vz) circle(1.3pt) node[above right,perspective] {$z$};

\draw[thick] (v0) -- node[perspective,below]{$1$} (v1);
\draw[thick] (v0) -- node[perspective,above left]{$|z|$} (vz);
\draw[thick] (vz) -- node[perspective,below right]{$|z-1|$} (v1);

\end{tikzpicture}
\caption{The three external vertices $z_0,z_1,z_2$ span the complex plane (picture by M. Borinsky).}
\label{fig:gfC}
\end{figure}

In the context of four-dimensional $\mathcal{N}=4$ Super-Yang-Mills (SYM) theory,
one encounters convergent four-point functions that are conformal in the sense that every internal vertex has weighted degree
\begin{equation}\label{Nv}
N_{x_i}^{\mathrm{conf}}=\sum_{e\sim x_i}\nu_e=D/\lambda,
\end{equation}
where the sum is over all edges $e=x_iy$ that are adjacent to the vertex $x_i$ (see e.g.\ \cite{SYM}). Whenever one has conformal symmetry, inversion of all coordinates can be
used to eliminate one external vertex. So, conformal four-point functions can also be expressed in terms of graphical functions. Explicitly, for the integral of a conformal graph
$G^{\mathrm{conf}}$ we obtain \cite{5twist},
\begin{align}\label{AGfc}
&A_{G^{\mathrm{conf}}}(z_0,z_1,z_2,z_3)\\\nonumber
&\quad=||z_{10}||^{\lambda(-N_0-N_1-N_2+N_3)}||z_{30}||^{\lambda(-N_0+N_1+N_2-N_3)}||z_{31}||^{\lambda(N_0-N_1+N_2-N_3)}||z_{32}||^{-2\lambda N_2}
f_{G^{\mathrm{conf}}\backslash\{z_3\}}^{(\lambda)}(z),
\end{align}
where $N_i=N_{z_i}^{\mathrm{conf}}$ for the external vertices $z_0,z_1,z_2,z_3$, and $z\in\CC$ corresponds to the invariants $||z_{ij}||^2$ by
\begin{equation}\label{inv4}
\frac{||z_{20}||^2||z_{31}||^2}{||z_{10}||^2||z_{32}||^2}=z\zz,\qquad \frac{||z_{21}||^2||z_{30}||^2}{||z_{10}||^2||z_{32}||^2}=(z-1)(\zz-1).
\end{equation}
Note that the constraint (\ref{Nv}) is only relevant in the context of $\mathcal{N}=4$ SYM theory and will not be assumed in the following.

The connection to the physically more relevant momentum space is by Fourier transforming the propagator $p_e$ ($\Gamma(x)=\int_0^\infty t^{x-1}\ee^{-t}\dd t$ is the gamma function)
\begin{equation}\label{Fourier1}
\frac1{||p||^{2\nu_e^p}}=\frac{\Gamma(\lambda\nu_e)}{2^{2\nu_e^p}\Gamma(\nu_e^p)}\int\frac{\dd^D x}{\pi^{D/2}}\frac{e^{\ii x\cdot p}}{||x||^{2\lambda\nu_e}}.
\end{equation}
The connection between the momentum space weight $\nu_e^p$ and the position weight space $\nu_e$ is given by
\begin{equation}\label{nunu}
\nu_e=\frac{\lambda+1-\nu_e^p}\lambda.
\end{equation}
Note that $\nu_e^p=1$ leads to $\nu_e=1$ in four dimensions ($\lambda=1$).

\begin{figure}
\momentumthreepoint\qquad\positionthreepoint
\caption{The momentum routing of a three-point integral $A_G^p(p_1,p_2)$ and the corresponding position space integral $A_G(0,p_1,p_2)$.
Note that the weights are connected via $\lambda\nu_e=\lambda+1-\nu_e^p$; see (\ref{nunu}). Except for the weights, the graphs $G^p$ and $G$ are identical.}
\label{fig:3pt}
\end{figure}

For the momentum space three-point integral of the graph $G^p$ in Figure \ref{fig:3pt} we obtain (see e.g.\ \cite{IZ})
\begin{equation}\label{AGAG}
A_G^p(p_1,p_2)=\frac{\Gamma(N_G^p)\prod_e\Gamma(\lambda\nu_e)}{\Gamma(\lambda N_G)\prod_e\Gamma(\nu_e^p)}\widehat{A_G}(0,p_1,p_2),
\end{equation}
where
\begin{equation}\label{Fourier}
\widehat{A_G}(0,p_1,p_2)=\int_{\RR^{2D}}\frac{\dd^D x_1\dd^D x_2}{(2\pi)^D}\ee^{\ii x_1\cdot p_2-\ii x_2\cdot p_1}A_G(0,x_1,x_2)
\end{equation}
is the Fourier transform of the three-point integral $A_G(0,x_1,x_2)$.
Moreover, in (\ref{AGAG}) we used the momentum space weight
\begin{equation}\label{eqNGp}
N_G^p=\Big(\sum_{e\in G}\nu_e^p\Big)-\frac D2h_G,
\end{equation}
where $h_G$ is the first Betti number (the number of independent cycles or the loop order) of the graph $G$.
The negative sign in the exponent on the right hand side of (\ref{Fourier}) stems from the orientation of the external momenta in Figure \ref{fig:3pt}.

Note that by graph homology, the number of vertices of $G$ and $h_G$ is connected to the number of edges of $G$ by
\begin{equation}\label{graphhom}
|\sV_G^{\mathrm{int}}|+3+h_G=|\sE_G|+1.
\end{equation}
From (\ref{nunu}) we obtain
\begin{equation}\label{NGNG}
N_G^p=\Big(\sum_{e\in G}\frac D2-\lambda\nu_e\Big)-\frac D2(|\sE_G|-|\sV_G^{\mathrm{int}}|-2)=D-\lambda N_G.
\end{equation}

It was proved in \cite{par} that the graphical function $f_G^{(\lambda)}(z)$ is invariant under taking the external dual of the graph $G$
if the weights are transformed by (\ref{nunu}) and $\lambda N_G=D/2$. This translates to the fact that the position space three-point function $A_G(0,z_1,z_2)$ is
invariant under the Fourier transform (\ref{Fourier}), $\widehat{A_G}=A_G$, if $G$ is planar.

Recently, it was noticed in the context of SYM theory by X. Jiang in \cite{Jiang} that in four dimensions the momentum space in integral
(\ref{AGAG}) can directly be expressed as position space integral, implying $\widehat{A_G}=A_G$ in four dimensions and if $N_G^p=D/2$.
(Jiang labels external vertices differently. The transition from his labels to the labels in Figure \ref{fig:3pt} is a relation for three-point integrals.)

Here, we show that in any (possibly noninteger) dimension and for any graph $G$ with $\lambda N_G=N_G^p=D/2$ (both conditions are equivalent by (\ref{NGNG})),
we have
\begin{equation}\label{AGhatAG}
\widehat{A_G}=A_G,\qquad \widehat{A_G^p}=A_G^p
\end{equation}
where, by (\ref{AGAG}), self-duality of the momentum space integral (the second equation) is equivalent to self-duality in position space.
To prove (\ref{AGhatAG}), we only need that $A_G$ can be expressed in terms of the invariants (\ref{inv3}).

Self-duality implies that any scalar massless momentum space three-point integral is directly (with no Fourier transformation) given by the graphical function
of the same graph with weights transformed according to (\ref{nunu})
\begin{equation}\label{AGpf}
A_G^p(p_1,p_2)=\frac{\prod_e\Gamma(\lambda\nu_e)}{\prod_e\Gamma(\nu_e^p)}\frac{f_G^{(\lambda)}(p)}{||p_1||^D},\quad\text{if}\;N_G^p=D/2.
\end{equation}
To derive the above equation we substituted (\ref{AGhatAG}) and (\ref{AGf}) into (\ref{AGAG}) where $\Gamma(N_G^p)/\Gamma(\lambda N_G)=1$ by (\ref{NGNG}).
The relation (\ref{inv3}) between invariants and $p\in\CC$ becomes
\begin{equation}\label{inv3p}
\frac{||p_2||^2}{||p_1||^2}=p\pp,\qquad \frac{||p_2-p_1||^2}{||p_1||^2}=(p-1)(\pp-1).
\end{equation}

Note that (\ref{AGpf}) extends to $N_G^p\neq D/2$ if one adds to $G^p$ an external leg of appropriate weight.
If this leg is attached to the vertex where $p_1$ is outgoing, then it gives a trivial factor $||p_1||^{-2\nu_e^p}$.
Consider, for example, a triangle of weight 1 edges in momentum space. This one loop diagram has $N_G^p=1$ in four dimensions.
With an extra leg of weight 1 attached, we obtain a graph $G'$ with $N_{G'}^p=2$. Hence,
\begin{equation}\label{ex2}
\thetreePp=\quad\thetreePx=\frac{4\ii D(p)}{||p_1||^4(p-\pp)}.
\end{equation}
For the triangle with $N_G^p=1$, we get $A_G^p(p_1,p_2)=4\ii D(p)/(||p_1||^2(p-\pp))$.

The article is organized as follows.
In Section \ref{sect2}, we prove (\ref{AGhatAG}) in integer dimensions as a lemma on Fourier transformation.

In Section \ref{sect3}, we switch to parametric space allowing noninteger dimensions and show that the result follows from Corollary 1.8 in \cite{par}.
We also give an alternative parametric representation for the graphical function of a graph $G$ that only uses graph polynomials.

In Section \ref{sect4}, we show that (\ref{AGpf}) implies a new identity for graphical functions.

Finally, in Section \ref{sect5}, we switch to $\phi^4$ theory and show that the new identity for graphical functions gives a new, powerful identity for scalar integrals
($\phi^4$ periods).

\section*{Acknowledgements}
The author is supported by the DFG-grant SCHN 1240/3-1.

\section{Integer dimensions}\label{sect2}
In this Section we prove (\ref{AGhatAG}) in integer dimensions $D\geq1$. To do this, we only use that
the Feynman integral $A_G$ can be expressed in terms of the invariants $||z_{20}||^2/||z_{10}||^2$ and $||z_{21}||^2/||z_{10}||^2$; see (\ref{AGf}) and (\ref{inv3}).
Note that in a scalar theory the Fourier transform is an involution, $\widehat{\widehat{A_G}}=A_G$.

\begin{lem}\label{lem1}
Let $1\leq D\in\ZZ$ and let $G$ be a connected weighted graph with three external vertices $0,z_1,z_2$ and $|\sV_G^{\mathrm{int}}|$ internal vertices
such that the position space Feynman integral $A_G(0,z_1,z_2)$ in (\ref{AG}) exists.
If $N_G=D/2$ (see (\ref{eqNG})), then $A_G(0,z_1,z_2)$ is invariant under the Fourier transform (\ref{Fourier}),
\begin{equation}\label{Fourier2}
A_G(0,z_1,z_2)=\int_{\RR^{2D}}\frac{\dd^D x_1\dd^D x_2}{(2\pi)^D}\ee^{\ii x_1\cdot z_2-\ii x_2\cdot z_1}A_G(0,x_1,x_2).
\end{equation}
\end{lem}
\begin{proof}
Up to a scale, we can express $A_G(0,x_1,x_2)$ in terms of the invariants $x_2^2/x_1^2$ and $(x_2-x_1)^2/x_1^2$
(see (\ref{AGf}) and (\ref{inv3}) for $x_0=0$),
$$
A_G(0,x_1,x_2)=\frac1{||x_1||^D}f\Big(\frac{x_2^2}{x_1^2},\frac{(x_2-x_1)^2}{x_1^2}\Big),
$$
for some function $f$. We first consider the integral over $x_2$,
$$
I_1=\int_{\RR^D}\frac{\dd^D x_2}{||x_1||^D}\ee^{-\ii x_2\cdot z_1}f\Big(\frac{x_2^2}{x_1^2},\frac{x_2^2-2x_2\cdot x_1+x_1^2}{x_1^2}\Big).
$$
We use polar decompositions for $x_1$ and $z_1$, $x_1=||x_1||\hat x_1$, $z_1=||z_1||\hat z_1$, where $\hat x_1$ and $\hat z_1$ are unit vectors.
By rotational symmetry, the integral $I_1$ only depends on $||x_1||$, $||z_1||$ and the scalar product $\hat x_1\cdot\hat z_1$.
In particular, $I_1$ is invariant under exchange of $\hat x_1$ and $\hat z_1$. Changing the integration variable $x_2$ to $||x_1||x_2$ yields
$$
I_1=\int_{\RR^D}\dd^D x_2\ee^{-\ii ||x_1||x_2\cdot \hat x_1||z_1||}f(x_2^2,x_2^2-2x_2\cdot\hat z_1+1).
$$
Note that the function $f$ does not depend on $x_1$ and that the exponent simplifies to $-\ii ||z_1||x_2\cdot x_1$.
By definition of the Fourier transform, the integrand has an infinitesimal Gaussian suppression factor $\ee^{-\sigma(x_1^1+x_2^2)}$ in the limit $\sigma\to0$.
With this convention, the integral is absolutely convergent and we may swap the order of the integrations. The integration over $x_1$ gives a $\delta$ function,
$$
\int_{\RR^D}\frac{\dd^D x_1}{(2\pi)^D}I_1\ee^{\ii x_1\cdot z_2}=\int_{\RR^D}\dd^D x_2\delta^D(z_2-||z_1||x_2)f(x_2^2,x_2^2-2x_2\cdot\hat z_1+1).
$$
The integral over $x_2$ gives $x_2=z_2/||z_1||$, yielding
$$
\frac1{||z_1||^D}f\Big(\frac{z_2^2}{z_1^2},\frac{z_2^2-2z_2\cdot z_1+z_1^2}{z_1^2}\Big)=A_G(0,z_1,z_2).
$$
\end{proof}

The first identity in (\ref{AGhatAG}) is proved in the Lemma. The second identity follows because, by (\ref{NGNG}), the momentum space integral $A_G^p(p_1,p_2)$ has the same scaling
weight as $A_G(0,z_1,z_2)$. Note that the factor of $\lambda$ in $\lambda N_G$ comes from the convention to extract $\lambda$ from the exponent in propagators (so that weight
one in position space corresponds to weight one in momentum space).

Alternatively, one can derive the second identity in (\ref{AGhatAG}) from the first identity by Fourier transforming (\ref{AGAG}).

Self-duality (\ref{AGhatAG}) is special to massless three-point functions. It does not generalize to higher $n$-point functions or to theories with particle masses.

\section{Parametric representation}\label{sect3}
It is a classic result in QFT that Feynman integrals can be expressed in parametric space (Schwinger or Feynman parameters) where the integration is over positive real variables
$\alpha_e$ that correspond to the edges $e\in\sE_G$ of a connected graph $G$; see e.g.\ \cite{IZ}. We define the graph (Kirchhoff, first Symanzik) polynomial \cite{KIR}
\begin{equation}\label{Psi}
\Psi_G(\alpha)=\sum_{\mathrm{span}\,T}\prod_{e\notin T}\alpha_e,
\end{equation}
where the sum is over spanning trees in $G$. The second Symanzik polynomial carries the kinematic structure (we only consider the massless case). It is a sum over spanning two-forests,
\begin{equation}\label{F}
F_G(\alpha)=\sum_{\mathrm{span}\,T_1\cup T_2}p(T_1)^2\prod_{e\notin T_1\cup T_2}\alpha_e,
\end{equation}
where
\begin{equation}\label{p}
p(T_1)=\sum_{v\in T_1^{\mathrm{ext}}}p(v)=-p(T_2).
\end{equation}
Here, the sum is over the vertices of $T_1$ that are external in $G$ and $p(v)$ is the ingoing momentum at $v$. By momentum conservation, the total ingoing momentum is zero,
hence $p(T_1)=-p(T_2)$.

With these polynomials, we can write the Feynman integral $A_G^p(p_1,\ldots,p_n)$ as projective integral over the positive coordinate simplex $\Delta=\{\alpha_e>0\}$,
\begin{equation}\label{AGp}
A_G^p(p_1,\ldots,p_n)=\frac{\Gamma(N_G^p)}{\prod_{e\in\sE_G}\Gamma(\nu_e^p)}\int_\Delta\Omega\frac{\prod_{e\in\sE_G}\alpha_e^{\nu_e^p-1}}{\Psi_G^{D/2-N_G^p}F_G^{N_G^p}}.
\end{equation}
The integration measure $\Omega$ is the projective volume form
\begin{equation}\label{Omega}
\Omega=\sum_{e\in\sE_G}(-1)^e\alpha_e\dd\alpha_1\wedge\ldots\wedge\dd\alpha_{e-1}\wedge\dd\alpha_{e+1}\wedge\ldots\wedge\dd\alpha_{|\sE_G|},
\end{equation}
where we assume that the edges are labeled from $1$ to $|\sE_G|$.

Note that in parametric representation, the dimension $D$ is a parameter in the exponents of the denominator. This makes it possible to generalize Feynman integrals
to any noninteger $D$ for which the integral (\ref{AGp}) exists. By analytic continuation, $A_G^p$ can be extended to a meromorphic function in $D$ on the complex plane.

In practice, one transfers to affine space by setting one $\alpha_e=1$ (or, less efficiently, $\sum_e\alpha_e=1$) yielding the integration
$\int_0^\infty\cdots\int_0^\infty\cdots\dd\alpha_1\wedge\ldots\wedge\dd\alpha_{e-1}\wedge\dd\alpha_{e+1}\wedge\ldots\wedge\dd\alpha_{|\sE_G|}$.
The sign $(-1)^e$ is absorbed by the orientation of the hyperplane $\{\alpha_e=1\}$.
For small to medium loop orders one can use the parametric representation to explicitly perform one integral after the other \cite{BogInt,BInt,HyperFORM,HyperInt}.
For higher loop orders, parametric integration becomes very slow, producing a large number of terms in intermediate steps. In the setup of massless three-point functions
(or conformal four-point functions), the method of graphical functions is much more efficient.

In \cite{par}, the analogous result in position space was derived. We need the following set of polynomials
\begin{equation}\label{Psi1}
\Psi_G^{1,\ldots,|\sV_G^{\mathrm{ext}}|}(\alpha)=\sum_{\mathrm{span}\,F}\prod_{e\notin F}\alpha_e,
\end{equation}
where the sum is over spanning $|\sV_G^{\mathrm{ext}}|$-forests in which every tree has exactly one external vertex, and
\begin{equation}\label{Psi2}
\Psi_G^{ij,\{k\neq ij\}}(\alpha)=\sum_{\mathrm{span}\,F}\prod_{e\notin F}\alpha_e,
\end{equation}
where the sum is over spanning $|\sV_G^{\mathrm{ext}}|-1$-forests in which one tree has the external vertices $i,j$ and every other tree has exactly one of the other external vertices.
From the second set of polynomials we construct the polynomial $\Phi$ that contains the position space invariants $||z_{ij}||^2$,
\begin{equation}\label{Phi}
\Phi_G(\alpha)=\sum_{i>j}||z_{ij}||^2\Psi_G^{ij,\{k\neq ij\}}(\alpha).
\end{equation}
For the three-star on the left hand side of (\ref{3star}), for example, we get
\begin{equation}\label{exPsi}
\Psi_G^{0,1,z}(\alpha)=\alpha_1\alpha_2+\alpha_1\alpha_3+\alpha_2\alpha_3,\quad\Phi_G(\alpha)=||z_{21}||^2\alpha_1+||z_{20}||^2\alpha_2+||z_{10}||^2\alpha_3,
\end{equation}
where we labeled the edges attached to $0,1,z$ by $1,2,3$.

If all edge weights $\nu_e>0$, we obtain for the Feynman integral $A_G(z_0,\ldots,z_{|\sV_G^{\mathrm{ext}}|-1})$ (see Corollary 1.8 in \cite{par})
\begin{equation}\label{AGz}
A_G(z_0,\ldots,z_{|\sV_G^{\mathrm{ext}}|-1})=\frac{\Gamma(\lambda N_G)}{\prod_{e\in\sE_G}\Gamma(\lambda\nu_e)}\int_\Delta\Omega\frac{\prod_{e\in\sE_G}\alpha_e^{\lambda(1-\nu_e)}}
{(\Psi_G^{1,\ldots,|\sV_G^{\mathrm{ext}}|})^{D/2-\lambda N_G}\Phi_G^{\lambda N_G}}.
\end{equation}

If we specialize to three external vertices $z_0,z_1,z_2$, the polynomials $\Psi_G^{10,2}$, $\Psi_G^{20,1}$, $\Psi_G^{21,0}$ are given by sums over two-forests.
With the identification of external vertices with momenta in Figure \ref{fig:3pt}, we obtain
\begin{equation}\label{PhiF}
\Phi_G=F_G.
\end{equation}
Moreover, with (\ref{nunu}) and (\ref{NGNG}) we obtain
\begin{equation}\label{AGz1}
A_G(0,p_1,p_2)=\frac{\Gamma(D-N_G^p)}{\prod_{e\in\sE_G}\Gamma(\lambda\nu_e)}\int_\Delta\Omega\frac{\prod_{e\in\sE_G}\alpha_e^{\nu_e^p-1}}
{(\Psi_G^{0,p_1,p_2})^{N_G^p-D/2}F_G^{D-N_G^p}}.
\end{equation}
If $N_G^p=D/2$, the polynomials $\Psi_G$ in (\ref{AGp}) and $\Psi_G^{0,p_1,p_2}$ in (\ref{AGz1}) drop and we obtain
\begin{equation}\label{AGz2}
\frac{\prod_{e\in\sE_G}\Gamma(\lambda\nu_e)}{\Gamma(D/2)}A_G(0,p_1,p_2)=\frac{\prod_{e\in\sE_G}\Gamma(\nu_e^p)}{\Gamma(D/2)}A_G^p(p_1,p_2).
\end{equation}
Equation (\ref{AGpf}) follows.

We close this section by expressing the polynomials $\Psi_G^{z_0,z_1,z_2}$ and $\Phi$ (see (\ref{Psi1}) and (\ref{Phi})) in terms of graph polynomials (\ref{Psi}).
The advantage of such a representation is that graph polynomials are well studied in literature; see e.g.\ \cite{BrH2,Sc2}.

To do this, we reduce the graph $G$ to the graph $G_{012}=G(z_0=z_1=z_2)$ by identifying all external vertices.
If we identify the edges in $G_{012}$ with the edges in $G$ we observe that
\begin{equation}\label{PsiPsi}
\Psi_G^{z_0,z_1,z_2}=\Psi_{G_{012}}
\end{equation}
because the spanning three-forests of $G$ are in one to one correspondence with the spanning trees of $G_{012}$.
This construction generalizes to any number of external vertices \cite{par}.

Likewise, for $i,j\in\{0,1,2\}$ we define $G_{ij}=G(z_i=z_j)$ as $G$ with vertices $z_i$ and $z_j$ identified.
So, $G_{ij}$ has the two external vertices $z_i=z_j$ and $z_k$, where $\{i,j,k\}=\{0,1,2\}$. Again, we identify the edges of $G_{ij}$ with the edges of $G$.
By correspondence of spanning two-forests with spanning trees, we obtain
\begin{equation}\label{PsiPsi2}
\Psi_G^{ik,j}+\Psi_G^{jk,i}=\Psi_{G_{ij}}.
\end{equation}
The above set of equations inverts to (this is only possible for three external vertices)
\begin{equation}\label{PsiPsi3}
\Psi_G^{ij,k}=\frac{\Psi_{G_{ik}}+\Psi_{G_{jk}}-\Psi_{G_{ij}}}2.
\end{equation}
Substitution into (\ref{Phi}) yields
\begin{align}\label{Phi1}
\Phi(\alpha)&=(z_1-z_0)^2\frac{\Psi_{G_{20}}+\Psi_{G_{21}}-\Psi_{G_{10}}}2+\text{two cyclic terms}\\\nonumber
&=\frac{(z_2-z_0)^2+(z_2-z_1)^2-(z_1-z_0)^2}2\Psi_{G_{10}}+\text{two cyclic terms}\\\nonumber
&=(z_2-z_1)\cdot(z_2-z_0)\Psi_{G_{10}}+(z_1-z_2)\cdot(z_1-z_0)\Psi_{G_{20}}+(z_0-z_2)\cdot(z_0-z_1)\Psi_{G_{21}}\\\nonumber
&=\Big(\frac{z(\zz-1)+\zz(z-1)}2\Psi_{G_{10}}+\frac{2-z-\zz}2\Psi_{G_{20}}+\frac{z+\zz}2\Psi_{G_{21}}\Big)||z_{10}||^2.
\end{align}
In the last equation, we used the identification (\ref{inv3}) of invariants with expressions in $z\in\CC$.

\section{A new identity for graphical functions}\label{sect4}
In this section, we utilize the equivalence of momentum space three-point functions and position space three-point functions to derive
a new identity for graphical functions.

An external leg of weight $\nu^p$ attached to a graph $G^p$ gives a factor of $||p||^{-2\nu^p}$ where $p$ is the momentum that goes through the leg.
When we label the external momenta by $z_1$ and $z_2$ (because we want to switch to position space later), then we obtain

\begin{equation}\label{mom3pt}
\momentumthreepointza=||z_1||^{2\nu^p}\momentumthreepointzd=||z_2-z_1||^{2\nu^p}\momentumthreepointzb=||z_2||^{2\nu^p}\momentumthreepointzc.
\end{equation}

If $G^p$ on the left hand side of (\ref{mom3pt}) has weight $N_G^p=D/2-\nu^p$, then we can use the correspondence between momentum space and position space
in (\ref{AGAG}) and (\ref{AGhatAG}) to translate the last two identities into relations between position space integrals.
Note that the $\Gamma$ factors in (\ref{AGAG}) are identical for each graph, so that they can be dropped.
The weights transform according to (\ref{nunu}); the weight $\nu^p$, in particular, transforms to $\nu$. The factors $||\bullet||^{2\nu^p}$ in (\ref{mom3pt}) translate to
edges between the corresponding external vertices of weight $-\nu^p/\lambda=\nu-(\lambda+1)/\lambda$.
By (\ref{NGNG}), the duals of the three right hand graphs in (\ref{mom3pt}) have weight $(\lambda+1)/\lambda$
(note that the edge of weight $\nu$ attached to $G$ creates an extra internal vertex).
If we also include the external edge of weight $\nu-(\lambda+1)/\lambda$, the total weight becomes $N_G^{\mathrm{tot}}=\nu$. We obtain

\begin{equation}\label{pos3pt}
\positionthreepointzd=\positionthreepointzb=\positionthreepointzc\text{if }N_G^{\mathrm{tot}}=\nu.
\end{equation}
Note that the graph $G$ (the gray disc) is {\em not} rotated in the above equations. It stays fixed while the external structures rotate.

To lift the restriction $N_G^{\mathrm{tot}}=\nu$, we consider the leftmost graph in (\ref{pos3pt})
without the external edge between $z_1$ and $0$ but with the edge attached to $z_2$ (the leftmost graph in (\ref{gfid})) and call this graph $G$ (with weight $N_G$).
To use (\ref{pos3pt}), we add two edges between $z_1$ and $0$, one with weight $(\lambda+1)/\lambda-N_G$ and one with weight $\nu-(\lambda+1)/\lambda$, to obtain the graph $G'$.
The total weight of the two edges is $\nu-N_G$, so that $N_{G'}=\nu$. We can apply (\ref{pos3pt}) to $G'$ if we consider the edge between $z_1$ and $0$
with weight $(\lambda+1)/\lambda-N_G$ as a part of the gray disc. Finally, we multiply the result with $||z_1||^{2\lambda(\nu-N_G)}$, so that on the left hand side we are back to graph $G$.
This gives
\begin{equation}\label{gfid}
\graphicalfunctiona=\graphicalfunctionb=\graphicalfunctionc.
\end{equation}
With (\ref{inv3}) (where $z_0=0$) we can identify $z_1,z_2$ with $1,z\in\CC$ and obtain an identity for graphical functions.
However, Equations (\ref{gfid}) are stronger when considered as identities for three-point functions: they can be applied to subgraphs in larger ambient graphs
(note that all three graphs have total weight $N_G$).
We will use this property in the next section to derive a new identity for scalar integrals in four-dimensional $\phi^4$ theory.

If, in the context of graphical functions, $\nu=1-k/\lambda$, with $0\leq k\in\ZZ$, then $f_G^{(\lambda)}(z)$ can be constructed from the graph $G$ with the edge at $z$
contracted (from the gray disc). In this case, Equations (\ref{gfid}) are not useful for the calculation of $f_G^{(\lambda)}(z)$.
In noninteger dimensions, however, there exist graphs with any weight $\nu$, where (\ref{gfid}) can be useful to calculate $f_G^{(\lambda)}(z)$.
Moreover, in any dimension, one can try to simplify a graphical function $f_G^{(\lambda)}(z)$ by substituting (\ref{gfid}) in a subgraph of $G$.

\section{A new identity for scalar integrals in $\phi^4$ theory}\label{sect5}

In this section, we consider (convergent) scalar integrals in massless $\phi^4$ theory, where all edges have weight $\nu_e=1$ and $D=4$, $\lambda=1$.
These scalar integrals (in a mathematical context called $\phi^4$ periods) are renormalization scheme independent contributions to the beta function arising
from vertex integrals with no sub-divergence. The corresponding graphs are called primitive (see e.g.\ \cite{Census}).
Alternatively, one can consider scalar integrals as massless $p$-integrals at unit momentum.

In parametric space, the scalar integral of a primitive $\phi^4$ graph $G$ can be calculated with the graph polynomial, see (\ref{Psi}),
\begin{equation}\label{para4d}
P_G=\int_\Delta\frac{\Omega}{\Psi_G(\alpha)^2}.
\end{equation}

In $\phi^4$ theory, there exists one primitive graph at one loop, no primitive graph at two loops, and
one primitive graph at three loops; see Figure \ref{fig:bubbleK4}.
The scalar integrals of the bubble and the tetrahedron are $1$ and $6\zeta(3)=6\sum_{k=1}^\infty k^{-3}$, respectively.

\begin{figure}
\begin{tikzpicture}[scale=0.6]
\begin{scope}[local bounding box=bubble]
	\coordinate (A) at (2,0.5);
	\coordinate (B) at (-2,0.5);
	\coordinate (C) at (-2,-0.5);
	\coordinate (D) at (2,-0.5);
    \draw[li,name path=li1] (A) .. controls (1,-1) and (-1,-1) .. (B);
    \draw[li,name path=li2] (C) .. controls (-1,1) and (1,1) .. (D);
    \fill[name intersections={of=li1 and li2}]
        (intersection-1) circle (3pt)
        (intersection-2) circle (3pt);
    \node[below=0.79 of bubble] {bubble};
\end{scope}

\begin{scope}[xshift=200,local bounding box=tetra]
	\coordinate (A) at (1.5,0);
	\coordinate (B) at (0,1.5);
	\coordinate (C) at (-1.5,0);
	\coordinate (D) at (0,-1.5);
	\coordinate (O) at (0,0);
    \draw[li] (A) -- (C);
	\draw[white,line width=8pt] (B) -- (D);
    \draw[li] (O) circle (1.5);
    \draw[li] (A) -- (2,0);
    \draw[li] (B) -- (0,2);
    \draw[li] (C) -- (-2,0);
    \draw[li] (D) -- (0,-2);
    \draw[li] (B) -- (D);
    \fill (A) circle (3pt);
    \fill (B) circle (3pt);
    \fill (C) circle (3pt);
    \fill (D) circle (3pt);
    \node[below=0.2 of tetra] {tetrahedron};
\end{scope}
\end{tikzpicture}
\caption{The bubble and the tetrahedron are the smallest primitive graphs in $\phi^4$ theory.}
\label{fig:bubbleK4}
\end{figure}

First systematic calculations of scalar integrals in $\phi^4$ theory were done with exact numerical methods by D. Broadhurst and D. Kreimer \cite{BK}.
The method was extended by the author in \cite{Census}.

With the theory of graphical functions \cite{gfe,gf,Shlog}, it was possible to extend the data in \cite{BK,Census} to hundreds of graphs up to loop order
eleven (and beyond \cite{ZZ}). In six-dimensional $\phi^3$ theory, results exist up to loop order nine. With this data, a connection to motivic Galois theory became visible \cite{coaction}
which led to further investigations of the motivic structure of QFTs \cite{Bcoact1,Bcoact2} (the `cosmic' Galois group) and of the geometries that underlie the number content
of $\phi^4$ periods \cite{BSmod,SchnetzFq,Sc2}.

It is convenient to use conformal symmetry and complete the primitive vertex graph $G$ by joining the four external legs to an extra vertex `$\infty$'.
The completion $\overline{G}$ of $G$ is a four-regular graph (every vertex has four edges) that serves as an equivalence class of decompleted graphs which all have the same
scalar integral \cite{Census}. Completion is the conformal analogon to the (projective) cut and glue identities for $p$-integrals that allows one to close an edge and open a
different edge.

We define
\begin{equation}\label{PG}
P_{\overline{G}}\equiv P_{\overline{G}\backslash\{v\}}
\end{equation}
for any vertex $v$ in $\overline{G}$.
Completion significantly reduces the number of graphs. A list of all completed primitive graph up to eight (decompleted) loops is in \cite{Census}.
Graphs up to eleven loops together with all known results are contained in the Maple package \cite{Shlog}.

Different completed graphs can have equal integrals. Some identities of such type are known: the twist \cite{gfe,gf}, the Fourier identity \cite{BK,gf},
the Fourier-split \cite{Furtherphi4}, and the five-twist \cite{5twist}.
However, even after many years, not all $\phi^4$ identities are known.

A practical solution comes from the combinatorial Hepp bound that is conjectured to be a graph invariant which detects equal scalar integrals \cite{EPHepp}.
More recently, the Martin sequence emerged as a tool that can unify all combinatorial invariants and that can also be used to determine the scalar integral \cite{PTalk,PYMartin}.

These recent developments support the picture that there exist many more identities for $\phi^4$ periods than those that can be explained by known transformations.
The first examples are at loop order eight, where it is conjectured that
\begin{equation}\label{conjid}
P_{8,30}=P_{8,36}\qquad\text{and}\qquad P_{8,31}=P_{8,35}
\end{equation}
in the notation of \cite{Census}.
Note that the graphs of the periods $P_{8,35}$ and $P_{8,36}$ do not transform under known identities.

Like every identity for graphical functions, Equations (\ref{gfid}) give identities for the scalar integral $P_G$ in the case that $G$ has a three-vertex split.
This follows from the property of the scalar integral that one can choose three vertices $0,1,z$ in $G$ and calculate $P_G$ from the graphical function of $G$ with external vertices
$0,1,z$ by integrating $z$ over the complex plane \cite{gf}.

The smallest example of (\ref{gfid}) in $\phi^4$ theory is
\begin{equation}
\grapha=\quad\graphb\quad=\quad\graphc,
\end{equation}
where in the middle graph one edge has weight $-1$ (as indicated) while all other edges have weight $1$.
The first and the last graph are identical. The middle graph can be obtained by a decompleted twist identity \cite{gf}.
In this case, Equations (\ref{gfid}) give no new results.

An example that is not given by the twist is
\begin{equation}
\graphd=\quad\graphe=\quad\graphf.
\end{equation}
If this identity is inserted into the graph $P_{7,1}$ (on the left hand side), we obtain
\begin{equation}
\graphg\quad=\graphi\quad=\quad\graphh.
\end{equation}
While the identity $P_{7,1}=P^{\mathrm{non\,}\phi^4}_{7,31}$ can also be obtained with the twist at different vertices, the identity $P_{7,1}=P^{\mathrm{non\,}\phi^4}_{7,18}$
is new. Note that, like every other identity on scalar integrals, the new identity does not change the loop order and hence generates a finite group action on all graphs with
convergent integral at a given loop order.

In general, the new identity is somewhat similar to the twist. Up to loop order eight, however, it is even a bit more powerful: the new identity generates 18 relations
that are not twists while only four twists are not also obtained with the new identity.

Further restricting to transformations that map $\phi^4$ graphs to $\phi^4$ graphs, we (e.g.) get the following identities at eleven loops that are not twists
\begin{equation}\label{Pids}
P_{11,2857}=P_{11,2868},\quad P_{11,2860}=P_{11,2863},\quad P_{11,3058}=P_{11,3061}.
\end{equation}
All three identities, however, follow from twists via the intermediate graphs $P_{11,2863},P_{11,2857},P_{11,3057}$, respectively. This means that
twists of $P_{11,2863}$ (e.g.) give $P_{11,2863}=P_{11,2857}$ and $P_{11,2863}=P_{11,2868}$ implying the left identity in (\ref{Pids}). We also conclude that $P_{11,2857}=P_{11,2860}$.

Altogether, the application of the new identity inside $\phi^4$ theory gives no new identities up to loop order eleven.
However, the new identity is more prone to provide identities between $\phi^4$ graphs and graphs outside $\phi^4$ theory, so that it should be significantly
more powerful in chains of transformations on all ($\phi^4$ and non $\phi^4$) scalar integrals (see \cite{5twist} for a more detail discussion in a similar situation).

Regretfully, even with such chains, the new identity does not explain the conjectured identities (\ref{conjid}).

The new identity can also be used (instead of the twist) in the setup of the five-twist \cite{5twist}, potentially leading to another new identity.
Experiments, however, have not produced any new results in $\phi^4$ theory.


\begin{thebibliography}{99}
\bibitem{BogInt} {\bf C. Bogner}, {\it MPL -- A program for computations with iterated integrals on moduli spaces of curves of genus zero}, Comput.\ Phys.\ Commun.\ 203, 339-353 (2016).
\bibitem{5lphi3} {\bf M. Borinsky, J.A. Gracey, M.V. Kompaniets, O. Schnetz}, {\it Five loop renormalization of $\phi^3$ theory with applications to the Lee-Yang edge
singularity and percolation theory}, Phys.\ Rev.\ D 103, 116024 (2021).
\bibitem{gfe} {\bf M. Borinsky, O. Schnetz}, {\it Graphical functions in even dimensions}, Comm.\ in Number Theory and Physics 16, No.\ 3, 515-614 (2022).
\bibitem{recursive} {\bf M. Borinsky, O. Schnetz}, {\it Recursive computation of Feynman periods}, JHEP 22(8), 291 (2022).
\bibitem{BK} {\bf D.J. Broadhurst, D. Kreimer}, {\it Knots and numbers in $\phi^4$ theory to 7 loops and beyond}, Int. J. Mod.\ Phys.\ C 6, 519 (1995).
\bibitem{BInt} {\bf F.C.S. Brown}, {\it The Massless Higher-Loop Two-Point Function}, Commun.\ Math.\ Phys.\ 287, 925-958 (2009).
\bibitem{BrH2} {\bf F.C.S. Brown}, {\it On the periods of some Feynman integrals}, arXiv:0910.0114 [math.AG] (2009).
\bibitem{Bcoact1} {\bf F.C.S. Brown}, {\it Feynman amplitudes, coaction principle, and cosmic Galois group}, Comm.\ in Number Theory and Physics 11, No.\ 3, 453-555 (2017).
\bibitem{Bcoact2} {\bf F.C.S. Brown}, {\it Notes on motivic periods}, Comm.\ in Number Theory and Physics 11, No.\ 3, 557-655 (2017).
\bibitem{BSmod} {\bf F.C.S. Brown, O. Schnetz}, {\it Modular forms in quantum field theory}, Comm.\ in Number Theory and Physics 7, No.\ 2, 293-325 (2013).
\bibitem{ZZ} {\bf F.C.S. Brown, O. Schnetz}, {\it Single-valued multiple polylogarithms and a proof of the zig-zag conjecture}, Jour.\ of Numb.\ Theory 148, 478-506 (2015).
\bibitem{SYM} {\bf J.M. Drummond, C. Duhr, P. Heslop, J. Pennington, V.A. Smirnov}, {\it Leading singularities and off-shell conformal integrals}, JHEP 8, 133-190 (2013).
\bibitem{par} {\bf M. Golz, E. Panzer, O. Schnetz}, {\it Graphical functions in parametric space}, Lett.\ Math.\ Phys.\ 107, No.\ 6, 1177-1182 (2017).
\bibitem{Furtherphi4} {\bf S. Hu, O. Schnetz, J. Shaw, K.A. Yeats}, {\it Further investigations into the graph theory of $\phi^4$-periods and the $c_2$ invariant},
Ann.\ Inst.\ H. Poincare D, Comb.\ Phys.\ Interact. 9, No.\ 3, 473-524 (2022).
\bibitem{IZ} {\bf J.C. Itzykson, J.B. Zuber}, {\it Quantum Field Theory}, Mc-Graw-Hill, (1980).
\bibitem{Jiang} {\bf X. Jiang} {\it Notes on Fourier transform and its application to three-point momentum-space integrals}, Phys.\ Rev.\ D 113, 036012 (2026).
\bibitem{HyperFORM} {\bf A. Kardos, S.-O. Moch, O. Schnetz} {\it HyperFORM -- a FORM package for parametric integration with hyperlogarithms}, :arXiv:2511.19992 [hep-ph] (2025).
\bibitem{KIR} {\bf G. Kirchhoff}, {\it Ueber die Aufl\"osung der Gleichungen, auf welche man bei der Untersuchung der linearen Vertheilung galvanischer Str\"ome gef\"uhrt wird},
Annalen der Physik und Chemie 72, No.\ 12, 497-508 (1847).
\bibitem{HyperInt} {\bf E. Panzer} {\it Algorithms for the symbolic integration of hyperlogarithms with applications to Feynman integrals}, Computer Physics Communications 188, 148-166 (2015).
\bibitem{EPHepp} {\bf E. Panzer}, {\it Hepp's bound for Feynman graphs and matroids}, Ann.\ Inst.\ H. Poincare D Comb.\ Phys.\ Interact.\ 10, No.\ 1, 31-119 (2022).
\bibitem{PTalk} {\bf E. Panzer}, {\it Combinatorial Feynman integrals and Ap\'ery}, talk presented at the `Workshop on combinatorics and algebraic geometry in QFT',
MPI, Bonn, Germany, 21st Aug.\ 2024.
\bibitem{coaction} {\bf E. Panzer, O. Schnetz}, {\it The Galois coaction on $\phi^4$ periods}, Comm.\ in Number Theory and Physics 11, No.\ 3, 657-705 (2017).
\bibitem{PYMartin} {\bf E. Panzer, K.A. Yeats} {\it Feynman symmetries of the Martin and $c_2$ invariants of regular graphs}, Comb.\ Theory 5, No.\ 1, \#10 (2025).
\bibitem{Census} {\bf O. Schnetz}, {\it Quantum periods: A census of $\phi^4$ transcendentals}, Comm.\ Number Theory and Physics 4, no.\ 1, 1-48 (2010).
\bibitem{SchnetzFq} {\bf O. Schnetz}, {\it Quantum field theory over $\FF_q$}, Electron.\ J. Comb.\ 18N1:P102 (2011).
\bibitem{gf} {\bf O. Schnetz}, {\it Graphical functions and single-valued multiple polylogarithms}, Comm.\ in Number Theory and Physics 8, No.\ 4, 589–675 (2014).
\bibitem{numfunct} {\bf O. Schnetz}, {\it Numbers and Functions in Quantum Field Theory}, Phys.\ Rev.\ D 97, 085018 (2018).
\bibitem{Sc2} {\bf O. Schnetz}, {\it Geometries in perturbative quantum field theory}, Comm.\ in Number Theory and Physics 15, No.\ 4, 743 – 791 (2021).
\bibitem{7loops} {\bf O. Schnetz}, {\it $\phi^4$ theory at seven loops}, Phys.\ Rev.\ D107, 036002 (2023).
\bibitem{6loops} {\bf O. Schnetz}, {\it $\phi^3$ theory at six loops}, Phys.\ Rev.\ D112, 016028 (2025).
\bibitem{gft} {\bf O. Schnetz}, {\it Graphical functions with spin}, JHEP No.\ 6, 53 (2025).
\bibitem{5twist} {\bf O. Schnetz}, {\it The five-twist identity for Feynman periods}, submitted to Comm.\ in Number Theory and Physics, arXiv:2505.02578 [hep.th] (2025).
\bibitem{Shlog} {\bf O. Schnetz}, {\tt HyperlogProcedures}, V0.8, Maple package {\tt https://github.com/oliverschnetz/HyperlogProcedures} (2025).
\bibitem{Zagierdilog} {\bf D. Zagier}, {\it The dilogarithm function}, Frontiers in Number Theory, Physics, and Geometry II, 3-65 (2007).
\end{thebibliography}
\end{document}